\documentclass[a4paper, preprint]{elsarticle}
\usepackage[T1]{fontenc}
\usepackage[utf8]{inputenc}
\usepackage[english]{babel}
\usepackage{amsmath}
\usepackage{amssymb}
\usepackage{amsfonts}
\usepackage{graphicx}
\usepackage{ragged2e}
\usepackage[margin=2.5cm]{geometry}

\usepackage[absolute,overlay]{textpos}

\usepackage{color}

\newtheorem{theorem}{Theorem}[section]
\newtheorem{corollary}[theorem]{Corollary}
\newenvironment{remark*}[1][Remark]{\begin{trivlist}
\item[\hskip \labelsep {\bfseries #1}]}{\nopagebreak\flushright$\square$\end{trivlist}}
\newenvironment{proof}[1][Proof]{\begin{trivlist}
\item[\hskip \labelsep {\bfseries #1}]}{\nopagebreak\flushright$\square$\end{trivlist}}

\allowdisplaybreaks

\newcommand{\bigtimes}{\mathop{\mbox{\LARGE$\boldsymbol\times$}}}

\renewcommand{\phi}{\varphi}

\renewcommand{\rho}{\varrho}

\renewcommand{\theta}{\vartheta}
\newcommand{\eps}{\ensuremath{\varepsilon}}

\renewcommand{\d}{\partial}

\newcommand{\fa}{\forall}

\newcommand{\lbetr}{\left\lvert}
\newcommand{\rbetr}{\right\lvert}

\newcommand{\betr}[1]{\lbetr {#1}\rbetr}

\renewcommand{\l}{\ensuremath{\left}}
\renewcommand{\r}{\ensuremath{\right}}\newcommand{\ubr}{\underbrace}

\newcommand{\opn}{\operatorname}

\newcommand{\then}{\ensuremath{\Rightarrow}}

\newcommand{\Df}{\ensuremath{\mathfrak{D}}}

\newcommand{\nn}[1][{}]{\ensuremath{\mathbb{N}^{#1}}}
\newcommand{\rn}[1][{}]{\ensuremath{\mathbb{R}^{#1}}}
\newcommand{\cn}[1][{}]{\ensuremath{\mathbb{C}^{#1}}}
\newcommand{\zn}[1][{}]{\ensuremath{\mathbb{Z}^{#1}}}
\newcommand{\Arsinh}[1][{}]{\ensuremath{\mathrm{Ar\ sinh}}}
\newcommand{\tr}[1][{}]{\ensuremath{\opn{tr}}}
\newcommand{\adj}[1][{}]{\ensuremath{\opn{adj}}}
\newcommand{\cof}[1][{}]{\ensuremath{\opn{cof}}}

\let\stdparagraph\paragraph
\renewcommand\paragraph{\vspace{1em}\stdparagraph}

\makeatletter                
\def\ps@pprintTitle{
 \let\@oddhead\@empty
 \let\@evenhead\@empty
 \def\@oddfoot{}%
 \let\@evenfoot\@oddfoot}
\makeatother

\begin{document}

\begin{frontmatter}

\title{Overcoming the sign problem in 1-dimensional QCD by new integration 
rules with polynomial exactness}

\renewcommand{\thefootnote}{\fnsymbol{footnote}}

\author[d]{A.~Ammon}
\ead{andreas.ammon@desy.de}

\author[b]{T.~Hartung}
\ead{tobias.hartung@kcl.ac.uk}

\author[a]{K.~Jansen}
\ead{karl.jansen@desy.de}

\author[c]{H.~Le\"ovey}
\ead{leovey@math.hu-berlin.de}

\author[a]{J.~Volmer}
\ead{julia.volmer@desy.de}

\address[a]{NIC, DESY Zeuthen, Platanenallee 6, D-15738 Zeuthen, Germany}
\address[b]{Department of Mathematics, King's College London, Strand, London WC2R 2LS, United Kingdom }
\address[c]{Institut f\"ur Mathematik, Humboldt-Universit\"at zu Berlin, Unter den Linden 6, D-10099 Berlin}
\address[d]{IVU Traffic Technologies AG, Bundesallee 88, 12161 Berlin, Germany}

\begin{textblock}{15}(12.2,0.5)
\setlength{\parindent}{0cm}
DESY 16-129
\end{textblock}

\begin{abstract}
  In this paper we describe a new integration method for the groups
  $U(N)$ and $SU(N)$, for which we verified numerically that it is
  polynomially exact for $N\le 3$. The method is applied to the
  example of 1-dimensional QCD with a chemical potential. We explore,
  in particular, regions of the parameter space in which the sign
  problem appears due the presence of the chemical potential. While
  Markov Chain Monte Carlo fails in this region, our new integration
  method still provides results for the chiral condensate on arbitrary
  precision, demonstrating clearly that it overcomes the sign problem.
  Furthermore, we demonstrate that our new method leads to orders of
  magnitude reduced errors also in other regions of parameter space.
\end{abstract}


\end{frontmatter}


\section{Introduction}

The sign problem in models of statistical and high energy physics
constitutes one of the greatest challenges for computational sciences,
because of the difficulty to evaluate such systems
\cite{Troyer:2004ge}. Many attempts using various techniques have been
developed but no general solution to overcome the sign problem has
been found so far \cite{Gattringer:2016kco}. On the other hand, the
sign problem appears in important problems in physics. For example, in
high energy physics, the sign problem prevents to fully understand the
physics of the early universe and to explain and interpret heavy ion
collisions. In order to progress with these questions, simulations
within the framework of lattice QCD with a non-zero chemical potential
would be required. However, these are impossible with present
techniques; see refs.~\cite{Sexty:2014dxa,Gattringer:2014nxa} for
recent reviews. The reason is that standard computations in lattice
QCD employ Markov Chain Monte Carlo (MC-MC) methods which need a
positiv integrand in order to be applicable. However, in the problem
just mentioned a chemical potential is required leading to a complex
integrand and therefore to an oscillating function. In particular, if
the sign cancellation errors due to the plural oscillations are of
significantly higher magnitude than the real integral value, it
becomes unfeasible to evaluate such systems.

Therefore, 
alternative approaches to MC-MC methods
need to be developed and 
in \cite{Jansen:2013jpa,Ammon:2015mra} we have proposed and tested 
Quasi Monte Carlo and iterated numerical integration
techniques. These methods can 
improve the convergence of the involved integrations and 
also have the potential to deal with the sign problem. 
However, 
in this paper we discuss yet another 
method of numerical integration for generic systems with a sign problem. 
This new method 
leads to an arbitrarily precise evaluation of the involved integrals and 
is based on a {\em complete symmetrization} of the integrals considered. 

This can be achieved through new integration rules on compact groups,
as developed in this article, which lead to polynomial exactness. We
test the method on the example of 1-dimensional QCD with a chemical
potential, see e.g. \cite{Ravagli:2007rw}, for which already other
approaches have been used to solve the sign problem \cite{Aarts2010}.
Although 1-dimensional QCD
is a model with an interest in its own as the strong coupling limit of
QCD \cite{Ilgenfritz:1984ff}, we consider it here only as a benchmark
model for testing our approach, especially since it is possible to
compute observables analytically and, thus, check the numerical
results directly. In particular, we will compute the chiral condensate
for a broad range of action parameters, including values of the
chemical potential that are impossible (for all practical purposes) to
address with standard Monte Carlo techniques.

The idea to symmetrize the involved integrals in a MC-MC 
simulation to achieve positivity 
and stable results 
has also been proposed in refs.~\cite{Bloch:2013ara,Bloch:2013qva}. 
However, in these works only an incomplete symmetrization has
been used and still a large number of Monte Carlo samples 
were necessary to obtain accurate results. 
In our approach, we perform a polynomially exact integration
avoiding the MC-MC step. This way,  
we only need a very small number of 
integration points. 
In fact, we can reach arbitrary (up to machine) precision 
for the targeted physical observables and 
avoid the MC error completely. 

For our computations, we employ the compact groups $U(N)$ and $SU(N)$
and give a description for a complete symmetrization for $N \le 3$. As
we will demonstrate, for these cases with our new approach the sign
problem is completely avoided.

This paper is composed in the following way: In section
\ref{sec:1dqcd}, we introduce the model of 1-dimensional QCD, show
analytic results of the partition function $Z$, and demonstrate the
difficulty to compute $Z$ for specific parameters numerically. In
section \ref{sec:symmetrization}, we describe the polynomially exact
method based on completely symmetrized spherical quadrature rules
\cite{Genz}. In section \ref{Numres}, we explain our numerical
computations in more detail, show results for the partition function
and the chiral condensate, and explain their behavior for different
parameter values. In section \ref{sec:conclusion}, we finally conclude
this paper.

\section{One dimensional lattice QCD}\label{sec:1dqcd}
Let us consider the following Dirac operator (cf., e.g., \cite{Ravagli:2007rw})
for a lattice with $n$ points 
\begin{align}
  \begin{aligned}
    \Df(U)
    =
    \begin{pmatrix}
      m&\frac{e^\mu}{2} U_{1}&&&&\frac{e^{-\mu}}{2} U_{n}^*\\
      -\frac{e^{-\mu}}{2} U_{1}^*&m&\frac{e^\mu}{2} U_{2}&&&\\
      &-\frac{e^{-\mu}}{2} U_{2}^*&m&\frac{e^\mu}{2} U_{3}&&\\
      &&\ddots&\ddots&\ddots&\\
      &&&-\frac{e^{-\mu}}{2} U_{n-2}^*&m&\frac{e^{-\mu}}{2} U_{n-1}\\
      -\frac{e^\mu}{2} U_{n}&&&&-\frac{e^{-\mu}}{2} U_{n-1}^*&m
    \end{pmatrix}
  \end{aligned}
\end{align}
where all empty entries are zero and the corresponding one flavor partition function 
\begin{align}
  &Z(m,\mu,G,n)=\int_{G^n}\det\Df(U)\ dh_G^n(U) \label{equ:partFunc}
\end{align}
where $G=U(N)$ or $G=SU(N)$, $N\in \nn$,  and $h_G$ is the corresponding (normalized) Haar measure on $G$.

In order to reduce the numerical effort in calculating $\det\Df$, we will first reduce the dimension using the following theorem. 

\begin{theorem}\label{det-reduction}
  Let $U_0:=U_n$, $\tilde m_1:=m$,
  \begin{align}
    \fa j\in[2,n-1]\cap\nn:\ \tilde m_j:=m+\frac{1}{4\tilde m_{j-1}},
  \end{align}
  and 
  \begin{align}
    \tilde m_n:=m+\frac{1}{4\tilde m_{n-1}}+\sum_{j=1}^{n-1}\frac{(-1)^{j+1}2^{-2j}}{\tilde m_j\prod_{k=1}^{j-1}\tilde m_k^2}.
  \end{align}  
  Then,
  \begin{align}
    \det\Df=&\det\l(\prod_{j=1}^{n}\tilde
    m_j+2^{-n}e^{-n\mu}\l(\prod_{j=0}^{n-1}U_{j}\r)^*+(-1)^{n}2^{-n}e^{n\mu}\prod_{j=0}^{n-1}U_{j}\r).
    \label{equ:detD}
  \end{align}
\end{theorem}
\begin{proof}
  appendix A 

\end{proof}
\begin{remark*}
  In particular, in the gauge satisfying $U_j=1$ except for $U_n=U$, Theorem \ref{det-reduction} yields
  \begin{align}
      \det\Df &=\det\l(\prod_{j=1}^{n}\tilde
      m_j+2^{-n}e^{-n\mu}U^*+(-1)^{n}2^{-n}e^{n\mu}U\r)      \label{equ:detDgauge}
      = \det\l(c_1 + c_2 U^*+ c_3 U\r),
  \end{align}
  with $c_1:=\prod_{j=1}^n\tilde m_j$, $c_2=2^{-n}e^{-n\mu}$, and
  $c_3=(-1)^n2^{-n}e^{n\mu}$.

  Mathematically speaking, \eqref{equ:detDgauge} is an application of
  ``Fubini''\footnote{Since all our groups are compact, they are
    unimodular and the Haar measures satisfy $h_{G\times H}=h_G\times
    h_H$ and $h_{G\rtimes H}=h_G\times h_H$ (cf., e.g., exercise 2.1.7
    in \cite{abbaspour-moskowitz}).} and translation invariance of
  the Haar measure since $\det\Df$ only depends on
  $\prod_{j=0}^{n-1}U_{j}$. We will frequently assume this form of
  $\Df$ in analytic computations and we have implemented this form of
  $\Df$ in order to reduce computational overhead. Similarly, $c_1$,
  $c_2$, and $c_3$ are standard notations in this paper. 
  Since $U \in U(N)$ or $U \in SU(N)$ $\det\Df$ is a polynomial of degree $N$.
\end{remark*}

As an observable of the model we, investigate the
  chiral condensate
\begin{align}
    \chi(m,\mu,G,n) = \d_m\ln
    Z(m,\mu,G,n)=&\frac{\d_mZ(m,\mu,G,n)}{Z(m,\mu,G,n)}=\frac{\int_G\d_m\det\Df\
      dh_G}{\int_G\det\Df\ dh_G}.
\end{align}
Since $\det\Df$ is a polynomial of degree $N$ and the derivative $\d_m$ only acts on the term $\prod_{j=1}^n\tilde m_j$ in Theorem \ref{det-reduction}, $\d_m\det\Df$ is still a polynomial of degree $N$ and $\d_m\prod_{j=1}^n\tilde m_j$ can be computed using symbolic differentiation.

Theorem \ref{det-reduction} not only allows us to reduce numerical
overhead but we can furthermore calculate the partition function \eqref{equ:partFunc}
(and therefore also the chiral condensate) analytically.
\begin{theorem}\label{int-analytic}
  Let $c_1:=\prod_{j=1}^n\tilde m_j$, $c_2=2^{-n}e^{-n\mu}$, and $c_3=(-1)^n2^{-n}e^{n\mu}$ with $\tilde m_j$ as in Theorem \ref{det-reduction}. Then,
  \begin{align}
    Z(m,\mu,U(1),n)=\int_{U(1)}\det\Df(U)\ dh_{U(1)}(U)=&c_1, \label{equ:Z_U1}
  \end{align}
  \begin{align}
    Z(m,\mu,U(2),n)=\int_{U(2)}\det\Df(U)\ dh_{U(2)}(U)=&c_1^2-c_2c_3,
  \end{align}
  \begin{align}
    Z(m,\mu,SU(2),n)=\int_{SU(2)}\det\Df(U)\ dh_{SU(2)}(U)=&c_1^2+c_2^2-c_2c_3+c_3^2,
  \end{align}
  \begin{align}
    Z(m,\mu,U(3),n)=\int_{U(3)}\det\Df(U)\ dh_{U(3)}(U)=c_1^3-2c_1c_2c_3,
  \end{align}
  and
  \begin{align}
    Z(m,\mu,SU(3),n)=\int_{SU(3)}\det\Df(U)\ dh_{SU(3)}(U)=&c_1^3-2c_1c_2c_3+c_2^3+c_3^3.
  \end{align}
\end{theorem}
\begin{proof}
  appendix B               

\end{proof}

In addition, we can deduce the behavior of $Z$ for $m\searrow0$.
\begin{corollary}\label{m-to-zero-limits}
  Let $\tilde m_1:=m$, $\tilde m_j:=m+\frac{1}{4\tilde m_{j-1}}$, $\tilde m_n:=m+\frac{1}{4\tilde m_{n-1}}+\sum_{j=1}^{n-1}\frac{(-1)^{j+1}4^{-j}}{m\prod_{k=1}^{j-1}\tilde m_k\tilde m_{k+1}}$, and $c_1:=\prod_{j=1}^n\tilde m_j$. Then,
  \begin{align}
    \lim_{m\searrow0}c_1=
    \begin{cases}
      2^{1-n}&,\ n\in 2\mathbb{N}\\
      0&,\ n\in2\mathbb{N}-1
    \end{cases}.
  \end{align}
  In particular,
  \begin{align}
    \lim_{m\searrow0}Z(m,\mu,U(1),n)=&
    \begin{cases}
      2^{1-n}&,\ n\in2\mathbb{N}\\
      0&,\ n\in2\mathbb{N}-1
    \end{cases},\\
    \lim_{m\searrow0}Z(m,\mu,U(2),n)=&
    \begin{cases}
      3\cdot2^{-2n}&,\ n\in2\mathbb{N}\\
      -2^{-2n}&,\ n\in2\mathbb{N}-1
    \end{cases},\\
    \lim_{m\searrow0}Z(m,\mu,SU(2),n)=&
    \begin{cases}
      3\cdot2^{-2n}+2^{1-2n}\cosh(2n\mu)&,\ n\in2\mathbb{N}\\
      2^{1-2n}\sinh(2n\mu)-2^{-2n}&,\ n\in2\mathbb{N}-1
    \end{cases},\\
    \lim_{m\searrow0}Z(m,\mu,U(3),n)=&
    \begin{cases}
      4\cdot2^{-3n}&,\ n\in2\mathbb{N}\\
      0&,\ n\in2\mathbb{N}-1
    \end{cases},\\
    \lim_{m\searrow0}Z(m,\mu,SU(3),n)=&
    \begin{cases}
      4\cdot2^{-3n}+2^{1-3n}\cosh(3n\mu)&,\ n\in2\mathbb{N}\\
      2^{1-3n}\sinh(3n\mu)&,\ n\in2\mathbb{N}-1
    \end{cases}.
  \end{align}
\end{corollary}
\begin{proof}
  appendix C                

\end{proof}

If $n\mu$ is large and $m$ small, we can see clearly why the
integrals in Theorem \ref{int-analytic} are difficult to treat numerically; especially the $U(N)$
cases. If we assume a stochastic approach, e.g., a Monte Carlo method,
then each evaluation of $\det\Df$ in the form \eqref{equ:detDgauge} is a value in the vicinity of $|c_2|^N + |c_3|^N \approx |c_3|^N = 2^{-Nn}e^{Nn\mu}$.\footnote{$|c_2|^N+ |c_3|^N = 2^{-Nn}e^{-Nn\mu} + |(-1)^{Nn}|2^{-Nn}e^{Nn\mu} \approx 2^{-Nn}e^{Nn\mu} = |c_3|^N$, due to the fact that $e^x > e^{-x}$ for $x \in \rn_{>0}$ and the (anti)symmetric shape of $e^x \pm e^{-x}$.} However, performing the integration (or taking the limit of infinitely many samples), there is a very high degree of cancellations to be observed. Since discrete Markov Chain Monte Carlo methods perform poorly with respect to such cancellations, they have to overcome an initial error in the vicinity of $e^{Nn\mu}$. In other words, as $n\mu$ grows larger, we need very good algorithms to suppress the initial error and the convergence 
\begin{align*}
  \mathrm{error}\approx\frac{\mathrm{constant}}{\sqrt{\mathrm{sample\ size}}}
\end{align*}
of Monte Carlo methods is simply not viable anymore. For example, in
Figure \ref{fig:MC-comp}, we compare a Monte Carlo method (using
re-weighting) to our new, polynomially exact method proposed in section
\ref{sec:symmetrization} (details of the numerical tests can be found
in section \ref{Numres}). The error bars, the known rate of
convergence $\frac{1}{\sqrt{\mathrm{sample\ size}}}$, and the here
seen relative error of order $1$ show that the Monte Carlo method
cannot reach the same level of precision with a reasonable number of
samples (note the different scales for the Monte Carlo and
polynomially exact results).

\begin{figure}[h]
  \centering
  \scalebox{.973}{
  \includegraphics[width=1\textwidth]{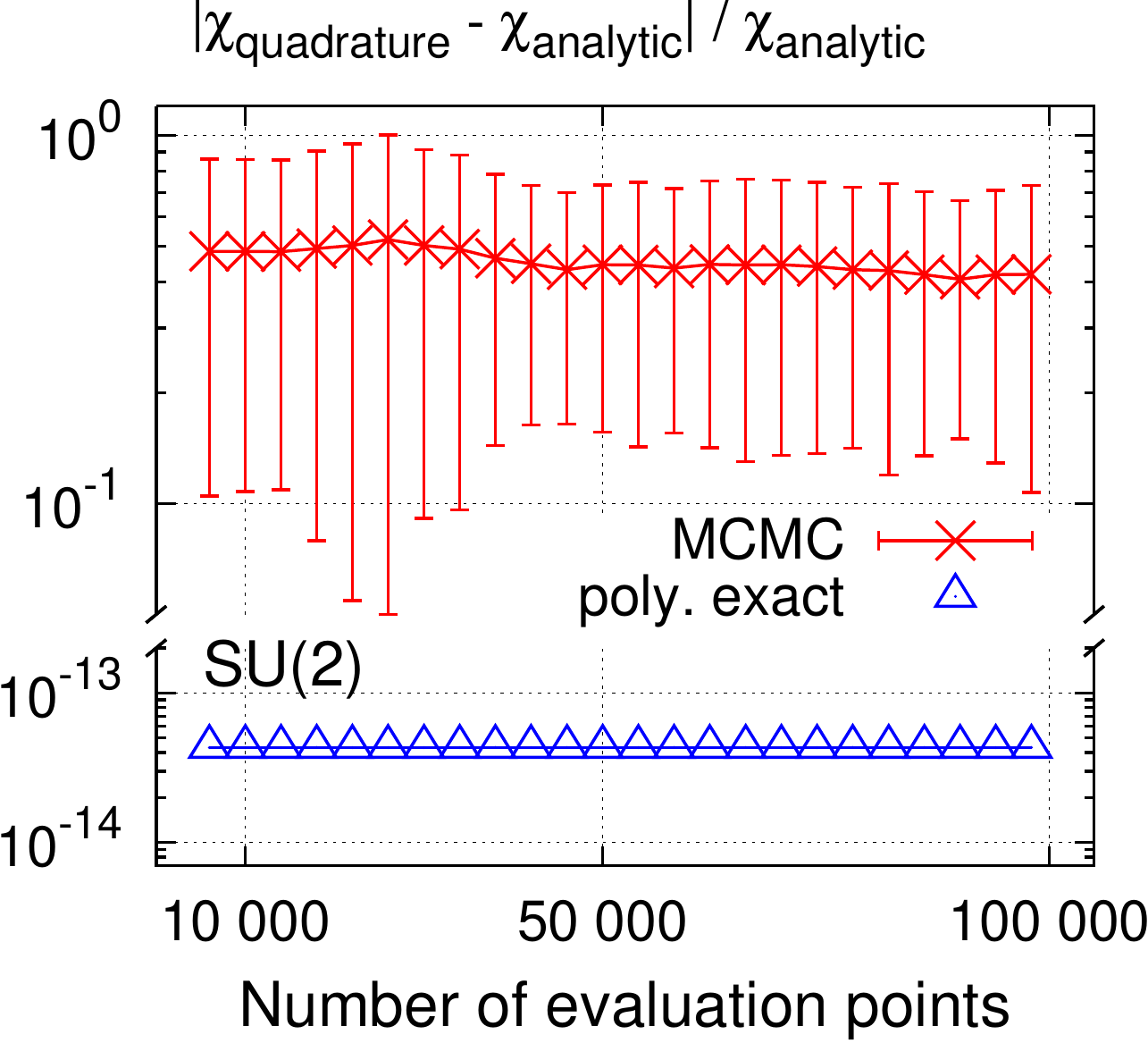}}
  \caption{Failure of MC-MC methods. Comparison of the relative error of
    the chiral condensate $\chi$ using polynomially exact (bottom) and Monte Carlo (top) quadrature rules for $SU(2)$. The polynomially exact rule used $n=8$ integration
    points, $m = 0.25$, $\mu = 1.0$, and the error bars have been computed from $20$ independent repetitions.}
  \label{fig:MC-comp}
\end{figure}

\section{Efficient quadrature rules over the compact groups}\label{sec:symmetrization}
Consider $Z(m,\mu,U(1),n)$ for the moment. As we have mentioned before, 
the problem is that the integral $\int_{U(1)}(-1)^{n}2^{-n}e^{n\mu}U\
dh_{U(1)}(U)$ in \eqref{equ:Z_U1} vanishes but the modulus of 
each evaluation $\betr{(-1)^{n}2^{-n}e^{n\mu}U}$ is large. However, if
we were also to evaluate at $-U$ 
(or, more generally at $t$ equally spaced points along the unit circle), the two terms would cancel. 
However, the (geometric) idea of taking opposite points or equally spaced points on circles, is not easy to formalize for $SU(N)$ and $U(N)$ with $N\ge2$. Instead, we should note that the quadrature rule
\begin{align}
 \int_{U(1)}f (U) dh_{U(1)}(U)
 \approx\frac{1}{t+1}\sum_{k=1}^{t+1}f\l(e^{\frac{2\pi ik}{t+1}}\r)
 \label{equ:tdesign_U1}
\end{align}
is a spherical $t$-design (i.e., an equal weights quadrature rule with 
spherical polynomial degree of exactness $t$; cf., Example 5.14 in \cite{delsarte}). 
Since $\det\Df$ is a polynomial of degree $N$ over the matrix entries for $U(N)$ and $SU(N)$, it suffices to consider $t$-designs 
or ``weighted'' $t$-designs (polynomially exact rules with possibly
non-equal weights) with $t=N$. 

In this section, we will discuss the construction of weighted $t$-designs for $N>1$ and, especially, why we base the $U(N)$ and $SU(N)$ quadrature on the quadrature
rules \cite{Genz} for the spheres $S^N$. 

Since
\begin{align}
  U(N)\cong SU(N)\rtimes U(1),
\end{align}
holds, where $\rtimes$ denotes the (outer) semi-direct product, we may construct a 
(weighted) $t$-design $Q_{U(N)}$ over $U(N)$ by considering two different (weighted) $t$-design rules $Q_{SU(N)}$ and $Q_{U(1)}$ over 
$SU(N)$ and $U(1)$ correspondingly, and then define the product rule $Q_{U(N)}=Q_{SU(N)}\times Q_{U(1)} $. It is clear that 
by defining $Q_{U(N)}$ as a product rule in this way, we obtain a (weighted) $t$-design over $U(N)$. Since $t$-designs over $U(1)$ are easy to 
construct (see \eqref{equ:tdesign_U1}), the entire problem of constructing (weighted) $t$-designs for the compact groups considered here reduces to the one of constructing 
(weighted) $t$-designs over $SU(N)$. 
 
Starting with $SU(2)$, we have a measure preserving diffeomorphism $SU(2)\cong S^3$. 
An explicit mapping can be given by
\begin{align}
  \Phi:\ \cn[2]\to\cn[2,2];\ (\alpha,\beta)\mapsto
  \begin{pmatrix}
    \alpha&-\beta^*\\
    \beta&\alpha^*
  \end{pmatrix}
\end{align}
whose restriction $\Phi|_{S^3}^{SU(2)}:\ S^3\to SU(2)$ is the mentioned measure preserving diffeomorphism. Thus, for this case we can resort to 
already well known (weighted) $t$-designs over the $3$-sphere (see \cite{Sloan2002227,Genz})  
for obtaining (weighted) $t$-designs over $SU(2)$ trough the mapping $\Phi$.\\

Moving on to $SU(3)$, we note that there is a correspondence\footnote{More precisely, $SU(N)$ is a principal $SU(N-1)$ bundle over $S^{2N-1}$; cf., e.g., \cite[equation (22.18)]{frankel}.} between $SU(3)$ and $S^{5} \times SU(2)$.
More specifically, we consider first the covering $\Phi_1: [0,2\pi)^{3}\times [0,\frac{\pi}{2})^{2}\rightarrow S^5$
defined by
\begin{align*}
x_1&= \cos(\alpha_1)\sin(\varphi_1)\\
x_2&= \sin(\alpha_1)\sin(\varphi_1)\\
x_3&= \sin(\alpha_2)\cos(\varphi_1)\sin(\varphi_2)\\
x_4&= \cos(\alpha_2)\cos(\varphi_1)\sin(\varphi_2)\\
x_5&= \sin(\alpha_3)\cos(\varphi_1)\cos(\varphi_2)\\
x_6&= \cos(\alpha_3)\cos(\varphi_1)\cos(\varphi_2)\\
\end{align*}
and note that the restriction $\Phi_1: [0,2\pi)^{3}\times (0,\frac{\pi}{2})^{2} \rightarrow  S^5_1$,
$S^5_1:=\Phi_1\left[[0,2\pi)^{3}\times (0,\frac{\pi}{2})^{2} \right]$, is a
diffeomorphism. Furthermore, the set $S^5_0:= S^5 \setminus S^5_1$ is a null set.  
On the other hand, we have the mapping 
$\Phi_2: \left([0,2\pi)^{3}\times [0,\frac{\pi}{2})^{2}\right) \times SU(2)\to  SU(3)$ defined by
\begin{align}
  \Phi_2((\alpha,\phi),U)=
  \begin{pmatrix}
    e^{i\alpha_1}\cos(\phi_1)&0&e^{i\alpha_1}\sin(\phi_1)\\
    -e^{i\alpha_2}\sin(\phi_1)\sin(\phi_2)&e^{-i\alpha_1-i\alpha_3}\cos(\phi_2)&e^{i\alpha_2}\cos(\phi_1)\sin(\phi_2)\\
    -e^{i\alpha_3}\sin(\phi_1)\cos(\phi_2)&-e^{-i\alpha_1-i\alpha_2}\sin(\phi_2)&e^{i\alpha_3}\cos(\phi_1)\cos(\phi_2)
  \end{pmatrix}
  \begin{pmatrix}
    U&0\\
    0&1
  \end{pmatrix},
\end{align}
whose restriction $\Phi_2: \left( [0,2\pi)^{3}\times (0,\frac{\pi}{2})^{2}\right) \times SU(2)  \rightarrow  SU(3)_1$ with
\begin{align}
  SU(3)_1:=\Phi_2 \left[ \left([0,2\pi)^{3}\times (0,\frac{\pi}{2})^{2}\right )\times SU(2)\right],
\end{align}
is a bijection and  the set $SU(3)_0:= SU(3) \setminus SU(3)_1$ is a Haar null set. Thus, starting with a (weighted) $t$-design rule 
$Q_{S^3}$ over $S^3$ and 
a (weighted) $t$-design $Q_{S_1^5}$ over $S^5$, such that each point of $Q_{S_1^5}$ lies in $S^5_1$, and considering the mapping
\begin{align}
  \Phi_3:\ S_1^5\times S^3\to SU(3);\ (x,y):=\Phi_2(\Phi_1^{-1}(x),\Phi(y)),
\end{align}
we obtain a quadrature rule $Q_{SU(3)}$ over $SU(3)$ by setting $Q_{SU(3)}:=\Phi_3\left[Q_{S_1^5}\times Q_{S^3}\right]$. 

In fact, by considering (randomized) fully symmetric interpolatory
rules $Q^{(1,3)}$ and $Q^{(1,5)}$ from \cite{Genz} as weighted
$t$-designs $Q_{S^3}$ and $Q_{S_1^5}$, we checked numerically that the
resulting quadrature rule $Q_{SU(3)}$ is again a weighted $t$-design
over $SU(3)$, for $t\le 3$. The latter observation drove us to
investigate a procedure more in detail for constructing weighted
$t$-design rules over $SU(N)$, for arbitrary positive integers $N$ and
$t$. This procedure is based on a generalization of the mapping
$\Phi_3$ as stated above and relies on the
correspondence\footnote{Induction over $SU(j)$ being a principal
  $SU(j-1)$ bundle over $S^{2j-1}$ \cite[equation (22.18)]{frankel}
  and $SU(2)\cong S^3$.} between $SU(N)$ and $\bigtimes_{j=1}^{N-1}
S^{2j+1}$. This new construction of quadrature rules over $SU(N)$ is
subject of current research by the authors, but the potential
applications of this new method exceed the scope of this article and
will be not reported at this point.

\section{Numerical results}\label{Numres}

In this section we will provide a comparison of the evaluation of the
partition function $Z$ and the chiral condensate $\chi$ using MC-MC
and our new polynomially exact integration rules. First we will concentrate on the partition function
$Z$. We will have a short loook at the behavior of the analytic
values of $Z$ before comparing them to the quadrature results of $Z$ using the
Monte Carlo and polynomially exact method in terms of a relative
error. To present the real power of the polynomially exact method,
we will show computational results for two different floating point number precisions. 
Then we will investigate the relative error behavior of the
chiral condensate. Since we compute
the relative error as the deviation of the quadrature result
from the computation using analytic formulae, we
explicitly differentiate these ways of computation in the following using the terms 
$Z_{\rm quadrature}$ and $Z_{\rm analytic}$.

As stated above, for the here considered model both, $Z$ and $\chi$ can be 
computed analytically for the groups $U(N)$ and $SU(N)$. In 
particular, the expression of the partition 
functions in Theorem \ref{int-analytic} for $SU(N)$ can be related to the one for $U(N)$ through
\begin{align}
  \begin{aligned}
    Z_{\rm analytic}(m,\mu,SU(N),n)=&Z_{\rm
      analytic}(m,\mu,U(N),n)+c_2^N+c_3^N\\
    =&
    Z_{\rm analytic}(m,\mu,U(N),n)+
    \begin{cases}
      2^{1-Nn}\cosh(Nn\mu)&,\ n\in2\nn\\
      -2^{1-Nn}\sinh(Nn\mu)&,\ n\in2\nn-1\; .
    \end{cases}
  \end{aligned}
  \label{equ:Zanalytic_SUN}
\end{align}
We note that for $U(N)$ the partition function smoothly approaches a much
smaller value than $c_2^N+c_3^N$ when decreasing the mass parameter
$m$ while for $SU(N)$ it approaches a constant near $c_2^N+c_3^N$ as given in
Theorem~\ref{int-analytic}, see also Corollary \ref{m-to-zero-limits}. The behavior of $Z_{\rm analytic}(m,\mu,G,n)$ as a function of the mass parameter $m$ for
$G\in\l\{U(3),SU(3)\r\}$, $n=6$, $\mu=1$, is shown in Figure
\ref{fig:comp_Z_exp} and there we can clearly see the different
behaviors of $Z_{\rm analytic}$ for $U(3)$ and $SU(3)$ for $m\searrow0$.

For the groups $U(N)$ and $SU(N)$ each point evaluation
 of the quadrature rule is of order
$O(2^{-Nn}e^{Nn\mu})$, that is, a double precision computation cannot
resolve values below $10^{-16}2^{-Nn}e^{Nn\mu}$. Since the behavior of
the partition function in comparison to the constant $c_2^N+c_3^N$
will be important in order to understand the relative error $|Z_{\rm quadrature}-Z_{\rm analytic}|/Z_{\rm analytic}$, we also show the value of
$|c_2^3+c_3^3| \approx 2^{-3n}e^{3n\mu}$ in Figure \ref{fig:comp_Z_exp}
(see discussion at the end of section
\ref{sec:1dqcd} above, as well) for the examples of $U(3)$ and $SU(3)$.
In Figure \ref{fig:comp_Z_exp}, we furthermore distinguished three regions with
different behavior, indicated in the following by region I, II and
III.

Let us first discuss the group $U(3)$. For large values of
$m$ (region III)
$2^{-3n}e^{3n\mu}$ is negligible compared to $Z_{\rm analytic}$. We
therefore expect a small deviation of $Z_{\rm quadrature}$ from $Z_{\rm
  analytic}$ and hence a small relative error. On the other hand, for small values
of $m$ (region I) $Z_{\rm analytic}$ becomes much
smaller than $2^{-3n}e^{3n\mu}$ and we expect a significant
relative error due to rounding errors. There is also a transition
regime in $m$ (region II) in which the values of $Z_{\rm analytic}$ and $2^{-3n}e^{3n\mu}$ have the same order of
magnitude. Hence, we expect a significant increase in the relative
error while decreasing $m$, but the smooth behavior of $Z_{\rm analytic}$
for $U(3)$ suggests that there will be a similarly smooth increase of
the relative error as a function of $m$. As we will
discuss below, this expectation is indeed verified in our numerical
tests.

In the $SU(3)$ case, we have the additional constant $c_2^3+c_3^3$
which, for $m$ small, is significantly larger than $Z_{\rm
  analytic}(U(3))$, see \eqref{equ:Zanalytic_SUN}. Looking at Figure
\ref{fig:comp_Z_exp}, we expect a relative error similar to the $U(3)$
case in region III. In region I, though, the relative error should be
much less now due to the fact that the analytic value and order of
magnitude of each point evaluation are closer together than in the
$U(3)$ case. In the transition region II, the behavior may be
different to $U(3)$ as well, although this is not deduced from the
figure per se but from the differences in the formulae of $Z_{\rm
  analytic}$ \eqref{equ:Zanalytic_SUN}. There, the $m$-dependent term of
$Z_{\rm analytic}$, the constant $c_2^N+c_3^N$, and the point
evaluation in the quadrature rules are of the same order of magnitude
$O(2^{-Nn}e^{Nn\mu})$. Thus, this additional term $c_2^N+c_3^N$, not
present at $U(N)$, could lead to competing effects for the relative
error and, hence, an irregular behavior of the relative error (at
least in the MC-MC case).

\begin{figure}[h!]
  \centering
  \scalebox{.973}{
  \includegraphics[width=1\textwidth]{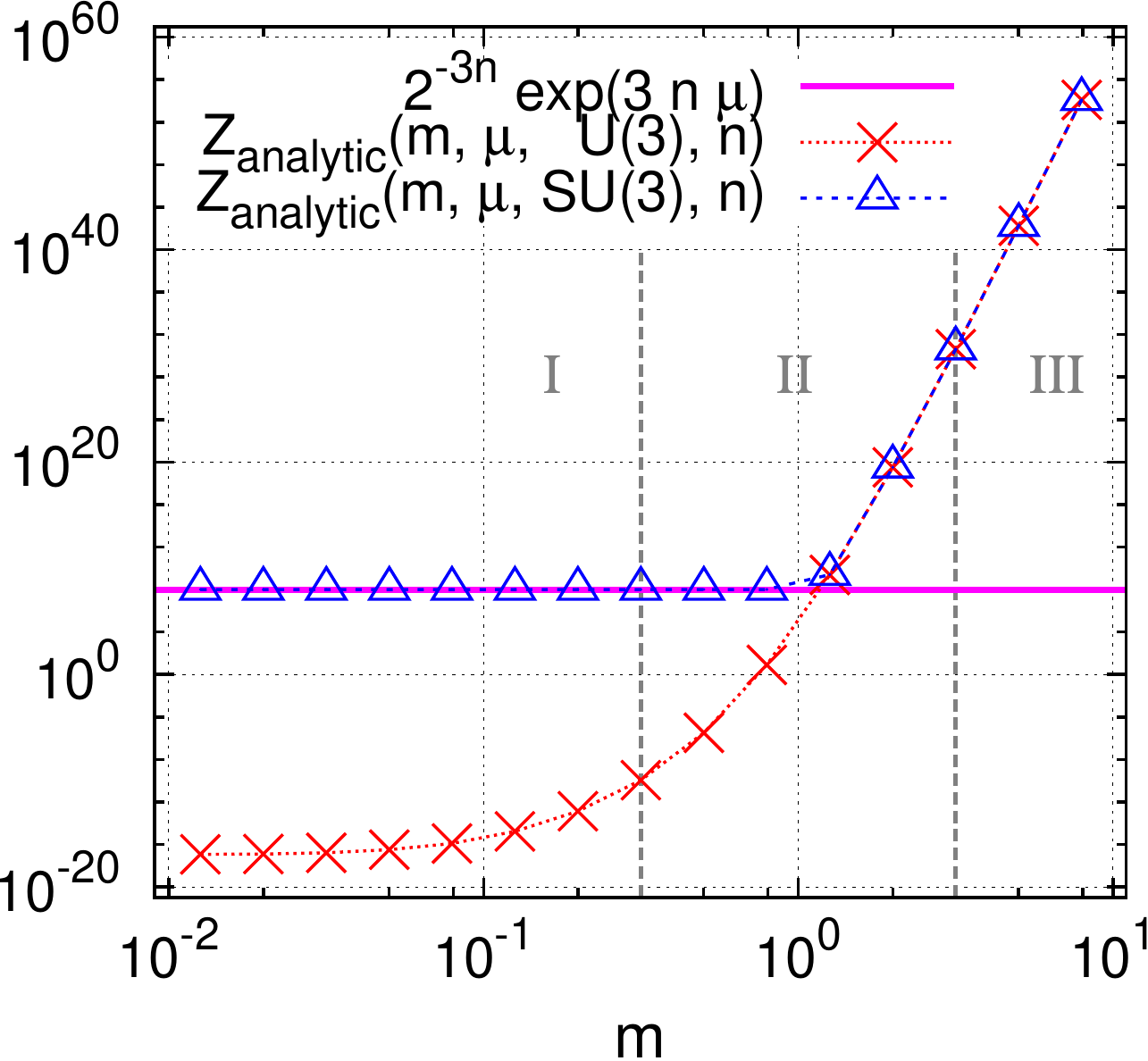}}
  \caption{Order of the quadrature rule point evaluation of the partition function
    integrand, $\left(2^{-n}e^{n\mu}\right)^3$, see \eqref{equ:detDgauge},
    compared to the analytic values of the partition functions
    for $U(3)$ and $SU(3)$ (see Theorem \ref{int-analytic}), using $n=20$, $\mu=1$. 
    As discussed in the paper, the ratio $Z_{\rm analytic} /
    2^{-3n}e^{3n\mu}$ determines the relative errors 
    of the partition function and the chiral condensate to a large extent. In particular, we identify three regions (I, II, III) in which the relative error exhibits qualitatively different behavior. (These computations were performed with $1024$bit floating point arithmetic.)}
  \label{fig:comp_Z_exp}
\end{figure}

Let us now move on to our numerical experiments. In Figure \ref{fig:comp}, we compare the quadrature rule
\begin{align}
  Z_{\rm quadrature}^{\rm MCMC}(m,\mu,G,n) = \int_G\det\Df\ dh_G\approx\frac{1}{\# Q_G}\sum_{k=1}^{\#
    Q_G}\det\Df(U_k) \label{equ:nonexactInt}
\end{align}
where each $U_k$ is chosen randomly in $G$ (uniformly with respect to the Haar measure) and the polynomially exact version
\begin{align}
  Z_{\rm quadrature}^{\rm poly.\ exact}(m,\mu,G,n) = \int_G\det\Df\ dh_G\approx\frac{1}{\# Q_G}\sum_{V\in
    Q_G}\det\Df(VU_1) \label{equ:exactInt}
\end{align}
where $U_1$ is the $U_1$ sampled in the non-exact version in \eqref{equ:nonexactInt}.\footnote{Any $U_1\in G$ would be perfectly fine; in fact, choosing the identity for $U_1$ would be a good canonical choice. However, we chose $U_1$ randomly (uniformly with respect to the Haar measure) in order to approximate the error.} Here, we chose
\begin{align}\label{eq:sym_groups}
  Q_G=
  \begin{cases}
    \l\{e^{\frac{2\pi ik}{4}};\ k\in\zn_4\r\}&;\ G=U(1)\\
    \Phi\left[Q_{S^3}\right]&;\ G=SU(2)\\
    \l\{e^{\frac{2\pi ik}{4}}U;\ k\in\zn_4,\ U\in \Phi\left[Q_{S^3}\right]\r\}&;\ G=U(2)\\
    \Phi_3\left[Q_{S_1^5}\times Q_{S^3}\right]&;\ G=SU(3)\\
    \l\{e^{\frac{2\pi ik}{4}}U;\ k\in\zn_4,\ U\in\Phi_3\left[Q_{S_1^5}\times Q_{S^3}\right]\r\}&;\ G=U(3)\\
  \end{cases}
\end{align}
where $Q_{S^3}$ and $Q_{S_1^5}$ are randomized fully symmetric 
rules of polynomial degree $3$ on $S^3$ and $S^5$ according to
\cite{Genz}. To obtain the error estimates, we repeated each numerical experiment $50$ times.
\begin{figure}[t!]
  \centering
  \includegraphics[width=1\textwidth]{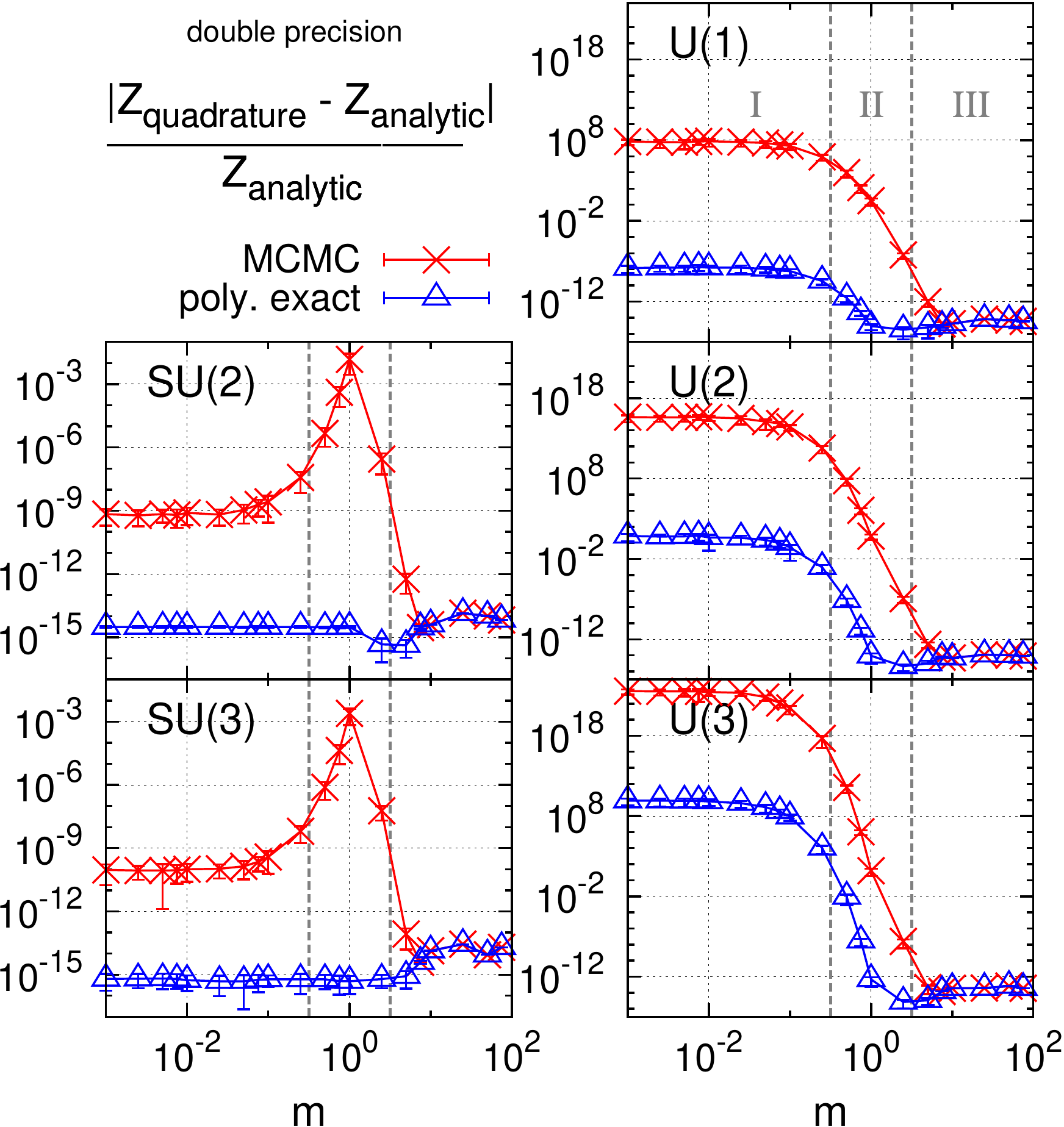}
  \caption{Comparison of the relative error of the used methods, namely the 
    polynomially exact and Monte Carlo quadrature rules to
    calculate the partition function $Z$ for
    $SU(2)$ and $SU(3)$ (left column, top to bottom), and $U(1)$,
    $U(2)$, and $U(3)$ (right column, top to bottom) with $n=20$,
    $\mu=1$, $m\in [0.001, 100]$. 
    Averages and standard deviations (error bars) have been
    computed from $50$ independent computations. Here we used double
    precision to carry through the numerical calculations. The different
    behaviors of the relative error regarding different values of $m$ 
    are divided into regions I,
    II and III, corresponding to Figure~\ref{fig:comp_Z_exp}.}
  \label{fig:comp}
\end{figure}

Figure \ref{fig:comp} shows the relative error of the partition function computed according to \eqref{equ:nonexactInt} and \eqref{equ:exactInt}. The same $m$-regions (I, II, and III), as shown in
  Figure \ref{fig:comp_Z_exp}, are indicated here as well and we can see
  that the behavior of the relative error is quite distinct
  for each of the three regions. For large values of $m$ (region III), both
methods operate with double precision as expected from 
the discussion above. 

Regarding regions I and II, we will consider the $U(N)$ case first. As
we move to smaller $m$, we
enter the transition region (II) and for $U(N)$ the relative error increases 
significantly but in a smooth way. 
As shown in Figure \ref{fig:comp_Z_exp},
for very small values of $m$ (region I) 
$Z_{\rm analytic}(m,\mu,U(N),n)$ is significantly smaller than $2^{-Nn}e^{Nn\mu}$;
hence, $Z_{\rm analytic}(m,\mu,U(N),n)$ is negligible 
compared to the machine error and we observe large relative 
errors in region I of Figure \ref{fig:comp}. Note that the polynomially 
exact computation still sums values of magnitude $2^{-Nn}e^{Nn\mu}$, 
i.e., the relative error of the exact method cannot be below $10^{-16}$ 
times the error of the non-exact method which is, indeed, what 
we see in Figure \ref{fig:comp}. Returning to 
Figure \ref{fig:comp_Z_exp} and the $U(N)$ discussion above, the
observed smooth increase of the relative error in region II
 matches our expectations.

In the $SU(N)$ case, the relative error is comparable to the
$U(N)$ case in regions I and III; we simply obtain smaller
errors in region I since $2^{-Nn}e^{Nn\mu}$ does not  dominate $Z_{\rm analytic}$ as is the case for $U(N)$. However,
in the transition region II of Figure~\ref{fig:comp}, we can see a rather irregular behavior whose possibility to occur we already mentioned
in the discussion of $Z_{\rm analytic}$ above. This can be
attributed to the fact that the mass dependent term of $Z_{\rm analytic}(m,\mu,U(N),n)$
and the constant $c_2^N+c_3^N$, see \eqref{equ:Zanalytic_SUN}, as well as the point evaluations in $Z_{\rm quadrature}$, are of the same order of magnitude. Hence neither term
can suppress the error of the other which we interpret as the cause of
the peak in the relative error.

Figure \ref{fig:comp-genz} shows the same comparison as 
Figure \ref{fig:comp} but computations were performed with 
$1024$bit
floating point arithmetic,\footnote{These are $1024$ mantissa bits; double precision (about $15$ digit precision) corresponds to $53$bit.} 
i.e., approximately $307$ digit precision. 
Again, we observe that the polynomially exact method 
operates on machine precision (as to be expected). 
The averages and standard deviations of the relative error 
were computed from $50$ independent computations for 
$G\in\{U(1),SU(2),U(2)\}$ and from $10$ independent computations 
for $G\in\{SU(3),U(3)\}$. 
All computations were performed on an IBM laptop in less than an 
hour.\footnote{Only for $SU(3)$ and $U(3)$, run-time was considerably longer than a few minutes.}
The behavior of the relative error, for both Monte-Carlo and the
polynomially exact method, is very similar to the double 
precision case in Figure~\ref{fig:comp}. Note that the polynomially 
exact integration always leads to machine precision results 
even in this extreme case of $1024$bit precision whereas the relative error of the MC-MC results does not notably decrease in regions
I and II when replacing double precision floats in Figure \ref{fig:comp} with 1024bit
extended floats in Figure \ref{fig:comp-genz}. 

\begin{figure}[t!]
  \centering
  \includegraphics[width=1\textwidth]{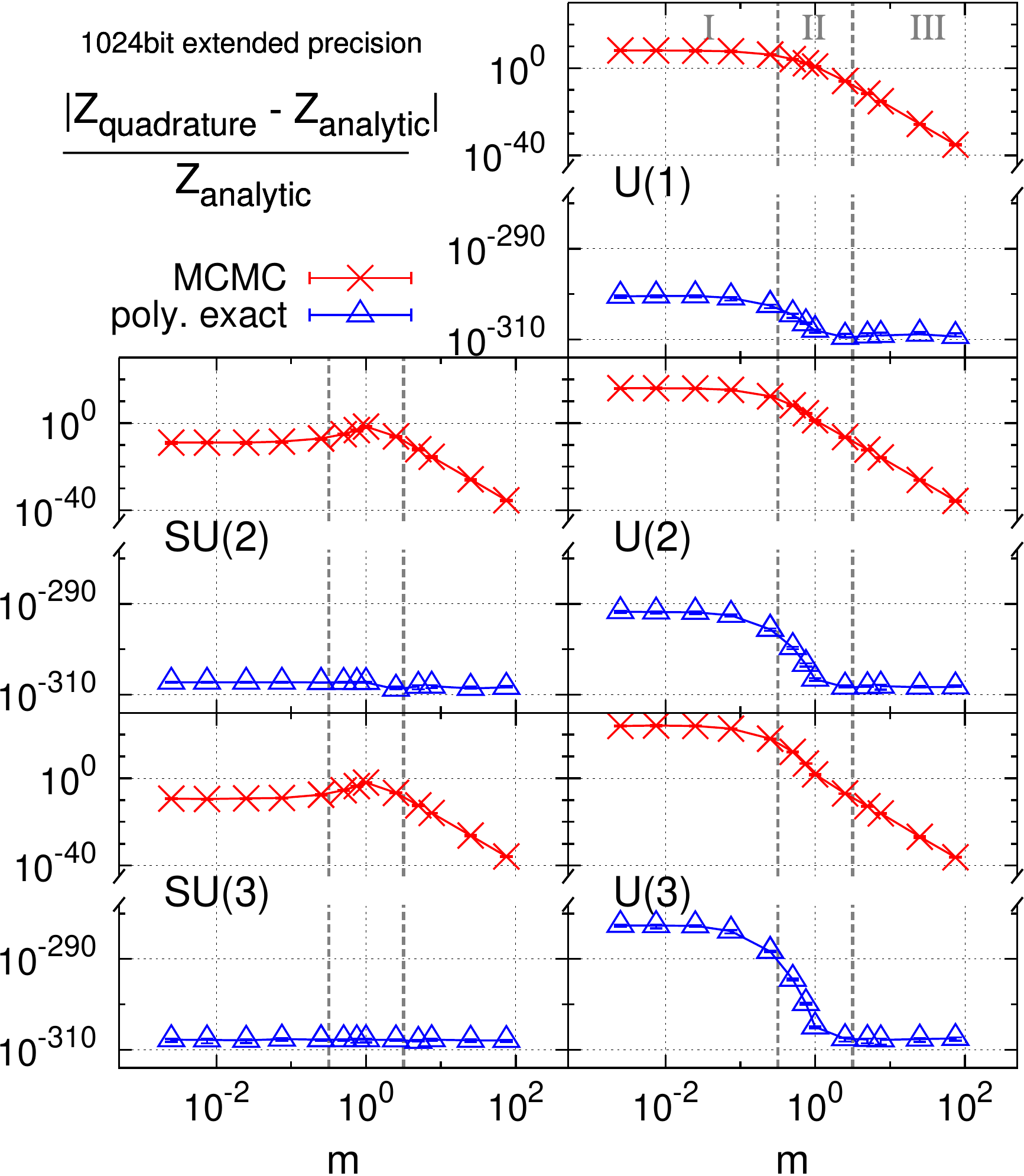}
  \caption{Comparison of the relative error as shown in Figure \ref{fig:comp}
    but here using $1024$bit extended floats. Averages and standard
    deviations (error bars) have been computed from $50$ independent
    computations for $U(1)$, $U(2)$, and $SU(2)$, and from $10$
    independent computations for $SU(3)$ and $U(3)$. Used are again $n=20$,
    $\mu=1$, $m\in [0.001, 100]$}
  \label{fig:comp-genz}
\end{figure}

In general, we observe 
in Figure~\ref{fig:comp} and Figure~\ref{fig:comp-genz} that the 
polynomially exact quadrature rule always provides machine error results.

In order to test our new polynomially exact method against 
an actual physical observable, Figure \ref{fig:comp-chiral} 
shows the comparison of the relative error of the chiral 
condensate (using $1024$bit extended floats again). The analytic values of the chiral condensate 
have been obtained through symbolic differentiation of the 
formulae in Theorem \ref{int-analytic}; the numerical values by symbolic differentiation of \eqref{equ:detDgauge}. 
We observe that the relative error follows the trend we have already seen for the partition function in the three different regions.\footnote{Here, the range of the regions differs from before.}

\begin{figure}[t!]
  \centering
  \includegraphics[width=1\textwidth]{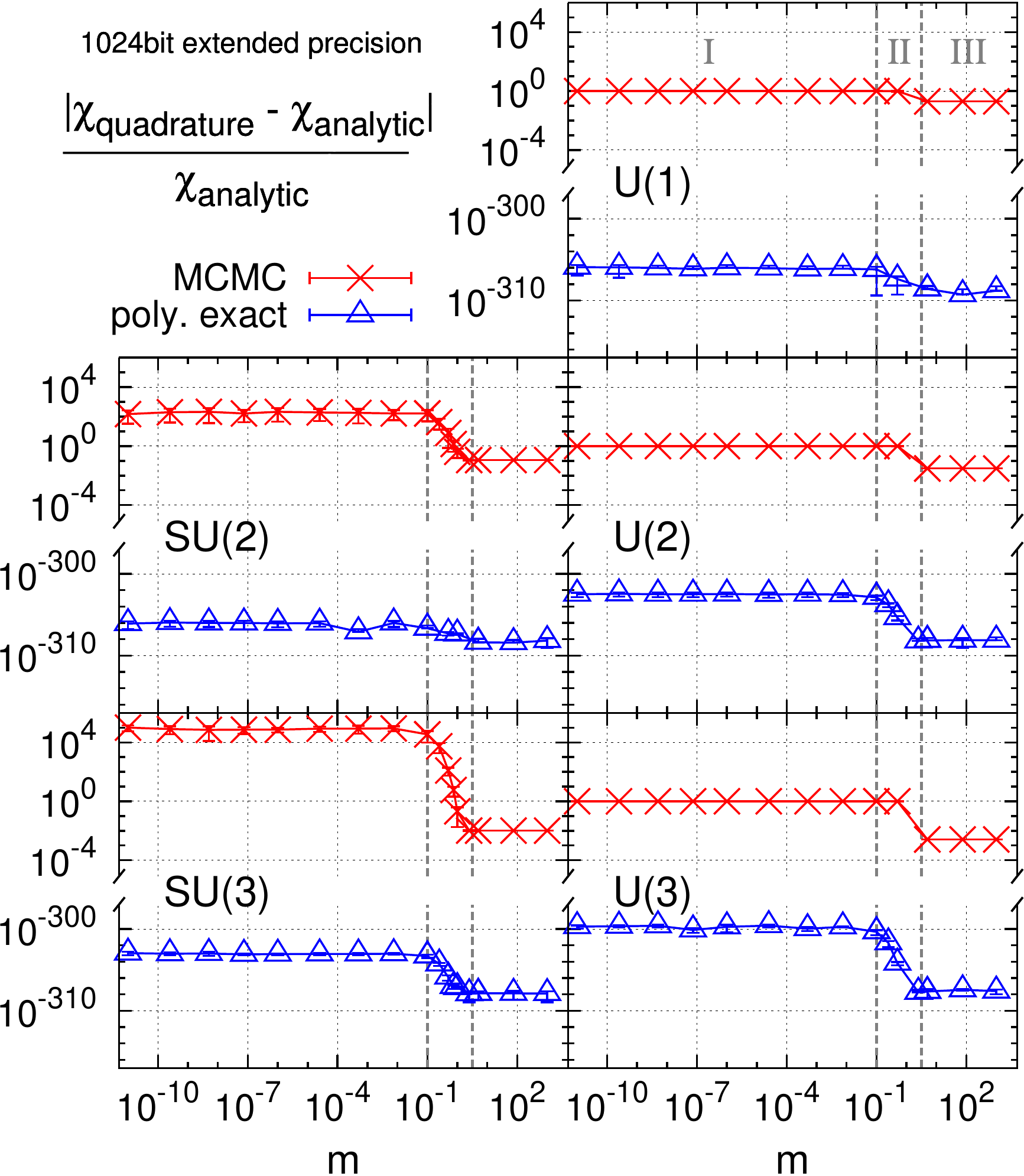}
  \caption{Comparison of
    the relative error of the chiral condensate $\chi = \partial_m \ln Z$ using polynomially exact and Monte Carlo
    quadrature rules for
    $SU(2)$ and $SU(3)$ (left column, top to bottom), and $U(1)$,
    $U(2)$, and $U(3)$ (right column, top to bottom) with $n=8$,
    $\mu=1.0$, and $m\in [10^{-11}, 10^3]$. 1024bit extended
    floats are used. Averages and standard deviations (error
    bars) have been computed from $50$ independent computations for
    $SU(2)$, $U(1)$ and $U(2)$ and from $5$ for $SU(3)$ and $U(3)$.
    The different
    behaviors of the error regarding different values of $m$ are divided into regions I,
    II and III.}
  \label{fig:comp-chiral}  
\end{figure}

Let us discuss the relative error in Figure~\ref{fig:comp-chiral} in a
bit more detail. A first observation is that the polynomially exact
method operates on the level of machine precision and, as such, reduces the relative error by (many) orders of magnitude for all
values of $m$. Even more interesting and striking is but the
size of the relative error of the chiral condensate in the small-$m$
region. As pointed out in \cite{Ravagli:2007rw}, in this region of
parameter space there is a severe sign problem. Indeed, for the MC-MC
method the relative error becomes $O(1)$ for sufficiently small $m$, i.e., no statistically significant result for the chiral
condensate can be obtained with standard
MC-MC calculations. (In Figure
\ref{fig:comp-chiral}, this behavior can only be seen for $U(N)$ but
it is also present and was observed by us for $SU(N)$ for $m$-values
smaller than the ones shown here.\footnote{The
  larger relative error $>O(1)$ for $SU(N)$ at
  small $m$ seen here is due to $\lim_{m\searrow0}\chi_{\rm analytic}(m,\mu,SU(N),n) = 0$, because
  $\lim_{m\searrow0} \partial_m Z_{\rm analytic}(m,\mu,SU(N),n) = 0$ and $\lim_{m\searrow0} Z_{\rm analytic}(m,\mu,SU(N),n) \neq 0$. 
  Thus, the analytic result for some small $m$ is already at machine
  precision while the quadrature result is not, such that division
  by this small machine precision number yields a value which can
  be larger than one.}) This is a clear manifestation of the
infamous sign problem.

In contrast, our polynomially exact method provides results on machine precision, again. Thus, the polynomially exact method completely overcomes
the sign problem and can lead to very accurate results even 
in regions where MC-MC computations are unfeasible.

\section{Conclusion}\label{sec:conclusion}

In this work, we have developed and tested a new 
integration method for the groups $U(N)$ and $SU(N)$. 
As a major outcome of our work, we could in fact provide a
numerical verification that 
the here developed method leads to polynomial exactness 
of the integration for $N\le 3$. 
We have applied the method 
to the 1-dimensional QCD with a chemical 
potential where for certain values
of the action parameters a sign problem appears with
  MC-MC methods. 
Using the groups $U(1),U(2),U(3)$ and $SU(2),SU(3)$
we have demonstrated that even for cases when the sign 
problem is most severe, the chiral condensate 
of this model can be computed to arbitrary precision
with the new method. In contrast, standard Markov Chain Monte Carlo 
methods show large $O(1)$ relative errors and do not give 
any statistically significant result. For this comparison, we even went
to 1024bit extended precision and were able to show that our new method
still achieves results on the level of machine precision. 
We, therefore, conclude that our polynomially exact integration method
can completely avoid the sign problem. 
Furthermore, it is 
important to point out that it also leads to orders of magnitude
reduced errors compared to MC-MC even in regions of parameter space 
where no sign problem occurs. 

The fact that our 
new integration method overcomes the sign problem and leads 
to orders of magnitude reduced errors in general in the here considered 
1-dimensional QCD is 
certainly a very promising finding and stands as a result by itself. 
However, this benchmark model can only 
be regarded as a toy example. It will be necessary to demonstrate that 
the method can also be applied in higher dimensions. To this end, 
we are presently considering the Schwinger model as an example
of a quantum field theory in 2 dimensions. 

Also, so far we do not have proof yet of the polynomial exactness
for the groups $U(N)$ and $SU(N)$ with general $N$.  
Although
we are very confident that our integration method leads to 
polynomial exactness for general $N$ we are working 
on a proof to substantiate this statement. 

\pagebreak
\section*{Acknowledgment}
The authors wish to express their gratitude to Prof. Andreas Griewank
for inspiring comments and
conversations, which helped to develop the work in this article.
H.L. and J.V. acknowledge financial support by the DFG-funded 
projects JA~674/6-1 and GR~705/13.


\bibliographystyle{unsrtnat}
\bibliography{Articles}

\pagebreak 

\appendix{\bf Appendix A: Proof of Theorem \ref{det-reduction}}\label{sec:proof-1}

Let 
\begin{align}
  Y=
  \begin{pmatrix}
    A&B\\C&D
  \end{pmatrix}
\end{align}
be a block decomposition where $A$ and $D$ are square matrices and $A$ is invertible. Then,
\begin{align}
  \det Y=\det A\det\l(D-CA^{-1}B\r).
\end{align}
Here, we are considering matrices of the form
\begin{align}
  X=
  \begin{pmatrix}
    m_1&\frac{e^{\mu}}{2} U_{1}&&&&\frac{e^{-\mu}}{2} U_{n}^*\\
    -\frac{e^{-\mu}}{2} U_{1}^*&m_2&\frac{e^{\mu}}{2} U_{2}&&&\\
    &-\frac{e^{-\mu}}{2} U_{2}^*&m_3&\frac{e^{\mu}}{2} U_{3}&&\\
    &&\ddots&\ddots&\ddots&\\
    &&&-\frac{e^{-\mu}}{2} U_{n-2}^*&m_{n-1}&\frac{e^{\mu}}{2} U_{n-1}\\
    -\frac{e^{\mu}}{2} U_{n}&&&&-\frac{e^{-\mu}}{2} U_{n-1}^*&m_n
  \end{pmatrix}
\end{align}
where all $m_i$ are positive. Choosing $A$ to be the $m_1$ block in $X$, we obtain
\begin{align}
  D=
  \begin{pmatrix}
    m_2&\frac{e^{\mu}}{2} U_{2}&&&&\\
    -\frac{e^{-\mu}}{2} U_{2}^*&m_3&\frac{e^{\mu}}{2} U_{3}&&&\\
    &-\frac{e^{-\mu}}{2} U_{3}^*&m_4&\frac{e^{\mu}}{2} U_{4}&&\\
    &&\ddots&\ddots&\ddots&\\
    &&&-\frac{e^{-\mu}}{2} U_{n-2}^*&m_{n-1}&\frac{e^{\mu}}{2} U_{n-1}\\
    &&&&-\frac{e^{-\mu}}{2} U_{n-1}^*&m_n
  \end{pmatrix}
\end{align}
and
\begin{align}
  -CA^{-1}B=&\frac{-1}{m_1}
  \begin{pmatrix}
    -\frac14&0&-\frac{e^{-2\mu}}{4}U_{1}^*U_{n}^*\\
    0&0&0\\
    -\frac{e^{2\mu}}{4}U_{n}U_{1}&0&-\frac14
  \end{pmatrix}.
\end{align}
In other words, $D-CA^{-1}B$ is of the initial form again and
\begin{align}
  \det X
  =&m_1^N\det(D-CA^{-1}B)\\
  =&\det
  \begin{pmatrix}
    m_2+\frac{1}{4m_1}&\frac{e^{\mu}}{2} U_{2}&&&&\frac{2^{-2}e^{-2\mu}}{m_1}U_{1}^*U_{n}^*\\
    -\frac{e^{-\mu}}{2} U_{2}^*&m_3&\frac{e^{\mu}}{2} U_{3}&&&\\
    &-\frac{e^{-\mu}}{2} U_{3}^*&m_4&\frac{e^{\mu}}{2} U_{4}&&\\
    &&\ddots&\ddots&\ddots&\\
    &&&-\frac{e^{-\mu}}{2} U_{n-2}^*&m_{n-1}&\frac{e^{\mu}}{2} U_{n-1}\\
    \frac{2^{-2}e^{2\mu}}{m_1}U_{n}U_{1}&&&&-\frac{e^{-\mu}}{2} U_{n-1}^*&m_n+\frac{1}{4m_1}
  \end{pmatrix}.
\end{align}
Let $U_0:=U_n$, $\tilde m_1:=m_1$,
\begin{align}
  \fa j\in[2,n-1]\cap\nn:\ \tilde m_j:=m_j+\frac{1}{4\tilde m_{j-1}},
\end{align}
and 
\begin{align}
  \tilde m_n:=m_n+\frac{1}{4\tilde m_{n-1}}+\sum_{j=1}^{n-1}\frac{(-1)^{j+1}2^{-2j}}{\tilde m_j\prod_{k=1}^{j-1}\tilde m_k^2}.
\end{align}
Then, we obtain inductively
\begin{align}
  &\det X\\
  =&\prod_{j=1}^{n-3}\tilde m_j^N\det
  \begin{pmatrix}
    \tilde m_{n-2}&\frac{e^{\mu}}{2} U_{n-2}&\frac{2^{-(n-2)}e^{-(n-2)\mu}}{\prod_{j=1}^{n-3}\tilde m_j}\l(\prod_{j=1}^{n-2}U_{j-1}\r)^*\\
    -\frac{e^{-\mu}}{2} U_{n-2}^*&m_{n-1}&\frac{e^{\mu}}{2} U_{n-1}\\
    \frac{(-1)^{n-2}2^{-(n-2)}e^{(n-2)\mu}}{\prod_{j=1}^{n-3}\tilde m_j}\prod_{j=1}^{n-2}U_{j-1}&-\frac{e^{-\mu}}{2} U_{n-1}^*&m_n+\sum_{j=1}^{n-3}\frac{(-1)^j2^{-2j}}{\tilde m_j\prod_{k=1}^{j-1}\tilde m_k^2}
  \end{pmatrix}\\
  =&\prod_{j=1}^{n-1}\tilde m_j^N\det\l(\tilde m_n+\frac{(-1)^{n}2^{-n}e^{n\mu}}{\prod_{j=1}^{n-1}\tilde m_j}\prod_{j=1}^{n}U_{j-1}+\frac{2^{-n}e^{-n\mu}}{\prod_{j=1}^{n-1}\tilde m_j}\l(\prod_{j=1}^{n}U_{j-1}\r)^*\r)
\end{align}
which finally yields
\begin{align}
  \det X=&\det\l(\prod_{j=1}^{n}\tilde m_j+(-1)^{n}2^{-n}e^{n\mu}\prod_{j=1}^{n}U_{j-1}+2^{-n}e^{-n\mu}\l(\prod_{j=1}^{n}U_{j-1}\r)^*\r).
\end{align}

\appendix{\bf Appendix B: Proof of Theorem \ref{int-analytic}}\label{sec:proof-2}

Note that the $U(1)$ case is trivial. Hence, we will start considering $U(N)$ with $N\ge2$ and use the notations
\begin{align}
  U_{ij}^*:=\l(U^*\r)_{ij}
\end{align}
and
\begin{align}
  \fa p\in\nn_0\ \fa I,J\in\nn[p]_{\le N}:\ U_{IJ}:=\prod_{k=0}^{p-1}U_{I_k J_k}\ \wedge\ U_{IJ}^*:=\prod_{k=0}^{p-1}\l(U^*\r)_{I_k J_k}.
\end{align}
Furthermore, we set $\fa p,q\in\nn_0\ \fa I,J\in \nn[p]_{\le N}\ \fa K,L\in \nn[q]_{\le N}:$
\begin{align}
   \langle I,J|K,L\rangle:=\int_{U(N)}U_{IJ}^*U_{KL}dh_{U(N)}(U)
\end{align}
and use abbreviations for empty sets or singletons similar to
\begin{align}
  \langle0,1|\rangle:=\langle(0),(1)|(),()\rangle.
\end{align}
The following identities are well-known (cf., e.g., \cite{Gattringer:2010zz}).
\begin{itemize}
\item $p\ne q\ \then\ \langle I,J|K,L\rangle=0$
\item $\langle|\rangle=1$
\item $\langle i,j|k,l\rangle=\frac{\delta_{il}\delta_{jk}}{N}$
\end{itemize}
For $N=2$, we may expand the determinant in
\begin{align}
  \int_{U(2)}\det\Df\ dh_{U(2)}
  =&\int_{U(2)}\det\l(c_1+c_2U^*+c_3U\r)dh_{U(2)}(U)\\
  =&\int_{U(2)}\det
  \begin{pmatrix}
    c_1+c_2U^*_{00}+c_3U_{00}&c_2U^*_{01}+c_3U_{01}\\
    c_2U^*_{10}+c_3U_{10}&c_1+c_2U^*_{11}+c_3U_{11}
  \end{pmatrix}dh_{U(2)}(U)
\end{align}
directly and, using the identities above, we obtain
\begin{align}
  \int_{U(2)}\det\Df\ dh_{U(2)}=&c_1^2-c_2c_3.
\end{align}
Similarly, we can expand the determinant in
\begin{align}
  \int_{U(3)}\det
  \begin{pmatrix}
    c_1+c_2U^*_{00}+c_3U_{00}&c_2U^*_{01}+c_3U_{01}&c_2U^*_{02}+c_3U_{02}\\
    c_2U^*_{10}+c_3U_{10}&c_1+c_2U^*_{11}+c_3U_{11}&c_2U^*_{12}+c_3U_{12}\\
    c_2U^*_{20}+c_3U_{20}&c_2U^*_{21}+c_3U_{21}&c_1+c_2U^*_{22}+c_3U_{22}
  \end{pmatrix}dh_{U(3)}(U)
\end{align}
using Sarrus' rule which yields (a few tedious pages later)
\begin{align}
  \int_{U(3)}\det\Df\ dh_{U(3)}
  =&c_1^3-2c_1c_2c_3
\end{align}
using the identities above.

For $SU(N)$, we have
\begin{itemize}
\item $p\ne q\ \then\ \langle I,J|K,L\rangle=0$
\item $\langle|\rangle=1$
\item $\langle i,j|k,l\rangle=\frac{\delta_{il}\delta_{jk}}{N}$
\item $\langle(i,j),(k,l)|\rangle=\langle|(i,j),(k,l)\rangle=-\frac{(-1)^{\delta_{ik}}}{2}=\frac{(-1)^{\eps_{ik}}}{2}$ in $SU(2)$
\item $\langle |(i,j,k),(l,m,n)\rangle=\langle (i,j,k),(l,m,n)|\rangle=\frac{\eps_{ijk}\eps_{lmn}}{6}$ in $SU(3)$
\end{itemize}
Hence, (analogous to the $U(N)$ computations)
\begin{align}
  \int_{SU(2)}\det\Df\ dh_{SU(2)}=&c_1^2+c_2^2-c_2c_3+c_3^2
\end{align}
and
\begin{align}
  \int_{SU(3)}\det\Df\ dh_{SU(3)}=&c_1^3-2c_1c_2c_3+c_2^3+c_3^3.
\end{align}

\appendix{\bf Appendix C: Proof of Corollary \ref{m-to-zero-limits}}\label{sec:proof-3}

By induction, we note for $2j<n$
\begin{align}
  \lim_{m\searrow0}\frac{\tilde m_{2j-1}}{jm}=&1\qquad\text{and}\qquad\lim_{m\searrow0}\frac{\tilde m_{2j}}{\quad\frac{1}{4jm}\quad}=1.
\end{align}
This is trivially true for $\tilde m_1=m$ and $\tilde m_2=m+\frac{1}{4m}$. Then, we observe for $j>1$
\begin{align}
  \lim_{m\searrow0}\frac{\tilde m_{2j-1}}{jm}=&\lim_{m\searrow0}\frac{m+\frac{1}{4\tilde m_{2j-2}}}{jm}
  =\lim_{m\searrow0}\frac{1}{j}+\frac{\quad\frac{1}{4\tilde m_{2j-2}}\quad}{jm}
  =\lim_{m\searrow0}\frac{1}{j}+\frac{\quad\frac{1}{4\frac{\tilde m_{2j-2}}{\quad\frac{1}{4(j-1)m}\quad}\frac{1}{4(j-1)m}}\quad}{jm}\\
  =&\frac{1}{j}+\frac{j-1}{j}
  =1
\end{align}
and
\begin{align}
  \lim_{m\searrow0}\frac{\tilde m_{2j}}{\quad\frac{1}{4jm}\quad}=&\lim_{m\searrow0}\frac{m+\frac{1}{4\tilde m_{2j-1}}}{\quad\frac{1}{4jm}\quad}
  =\lim_{m\searrow0}4jm^2+\frac{jm}{\tilde m_{2j-1}}
  =1.
\end{align}
Thus, we obtain 
\begin{align}
  \lim_{m\searrow0}\tilde m_k\tilde m_{k+1}=&
  \begin{cases}
    \lim_{m\searrow0}\frac{\tilde m_k}{jm}\frac{\tilde m_{k+1}}{\quad\frac{1}{4jm}\quad}\frac{jm}{4jm}&,\ k=2j-1\\
    \lim_{m\searrow0}\frac{\tilde m_k}{\quad\frac{1}{4jm}\quad}\frac{\tilde m_{k+1}}{(j+1)m}\frac{(j+1)m}{4jm}&,\ k=2j
  \end{cases}
  =
  \begin{cases}
    \frac{1}{4}&,\ k=2j-1\\
    \frac{j+1}{4j}&,\ k=2j
  \end{cases}
\end{align}
and for $n\in2\mathbb{N}$
\begin{align}
  \lim_{m\searrow0}c_1=&\lim_{m\searrow0}m\tilde m_n\prod_{j=1}^{\frac{n}{2}-1}\ubr{\tilde m_{2j}\tilde m_{2j+1}}_{\to\frac14\frac{j+1}{j}}\\
  =&2^{1-n}n\lim_{m\searrow0}m\tilde m_n\\
  =&2^{1-n}n\lim_{m\searrow0}m\left(m+\frac{1}{4\tilde m_{n-1}}+\sum_{j=1}^{n-1}\frac{(-1)^{j+1}4^{-j}}{m\prod_{k=1}^{j-1}\tilde m_k\tilde m_{k+1}}\right)\\
  =&2^{1-n}n\lim_{m\searrow0}\left(\frac{\frac{n}{2}m}{4\tilde m_{n-1}}\frac{2}{n}+\sum_{j=1}^{n-1}\frac{(-1)^{j+1}4^{-j}}{\prod_{k=1}^{j-1}\tilde m_k\tilde m_{k+1}}\right)\\
  =&2^{1-n}n\left(\frac{1}{2n}+\lim_{m\searrow0}\sum_{j=1}^{\frac{n}{2}}\frac{(-1)^{(2j-1)+1}4^{-(2j-1)}}{\prod_{k=1}^{(2j-1)-1}\tilde m_k\tilde m_{k+1}}+\sum_{j=1}^{\frac{n}{2}-1}\frac{(-1)^{2j+1}4^{-2j}}{\prod_{k=1}^{2j-1}\tilde m_k\tilde m_{k+1}}\right)\\
  =&2^{1-n}n\left(\frac{1}{2n}+\sum_{j=1}^{\frac{n}{2}}\frac{4^{1-2j}}{\prod_{k=1}^{2(j-1)}\lim_{m\searrow0}\tilde m_k\tilde m_{k+1}}-\sum_{j=1}^{\frac{n}{2}-1}\frac{4^{-2j}}{\prod_{k=1}^{2j-1}\lim_{m\searrow0}\tilde m_k\tilde m_{k+1}}\right)\\
  =&2^{1-n}n\left(\frac{1}{2n}+\sum_{j=1}^{\frac{n}{2}}\frac{4^{1-2j}}{4^{2-2j}\prod_{k=1}^{j-1}\frac{k+1}{k}}-\sum_{j=1}^{\frac{n}{2}-1}\frac{4^{-2j}}{4^{1-2j}\prod_{k=1}^{j-1}\frac{k+1}{k}}\right)\\
  =&2^{1-n}n\left(\frac{1}{2n}+\sum_{j=1}^{\frac{n}{2}}\frac{1}{4j}-\sum_{j=1}^{\frac{n}{2}-1}\frac{1}{4j}\right)\\
  =&2^{1-n}n\left(\frac{1}{2n}+\frac{1}{2n}\right)\\
  =&2^{1-n}.
\end{align}
Similarly, for $n\in 2\mathbb{N}-1$,
\begin{align}
  \lim_{m\searrow0}c_1=&\lim_{m\searrow0}\tilde m_n\prod_{j=1}^{\frac{n-1}{2}}\ubr{\tilde m_{2j-1}\tilde m_{2j}}_{\to\frac14}\\
  =&2^{1-n}\lim_{m\searrow0}\left(m+\frac{\quad4\frac{n-1}{2}m\quad}{4\frac{\tilde m_{n-1}}{\quad\frac{1}{4\frac{n-1}{2}m}\quad}}+\sum_{j=1}^{\frac{n-1}{2}}\frac{4^{1-2j}}{m\prod_{k=1}^{2j-2}\tilde m_k\tilde m_{k+1}}-\sum_{j=1}^{\frac{n-1}{2}}\frac{4^{-2j}}{m\prod_{k=1}^{2j-1}\tilde m_k\tilde m_{k+1}}\right)\\
  =&2^{1-n}\lim_{m\searrow0}\left(\sum_{j=1}^{\frac{n-1}{2}}\frac{4^{1-2j}}{4^{2-2j}jm}-\sum_{j=1}^{\frac{n-1}{2}}\frac{4^{-2j}}{4^{1-2j}jm}\right)\\
  =&0.
\end{align}
Finally, the asserted identities for $Z(m,\mu,G,n)$ with $G\in\{U(1),SU(2),U(2),SU(3),U(3)\}$ are a trivial corollary substituting $\lim_{m\searrow0}c_1$ into the formulae given in Theorem \ref{int-analytic}.

\end{document}